\documentclass{eptcs}
\providecommand{\event}{WRS 2009} % Name of the event you are submitting to
\providecommand{\volume}{??}
\providecommand{\anno}{20??}
\providecommand{\firstpage}{1}
\providecommand{\eid}{??}

\usepackage{amssymb}
\usepackage{amsmath}
\usepackage{amsthm}

\renewcommand{\email}[1]{\href{mailto:#1}{\footnotesize\tt #1}}

\hypersetup{
  pdfborder = 0 0 0.5,
  pdfdisplaydoctitle,
  pdftitle={Stream Productivity by Outermost Termination},
  pdfauthor={Hans Zantema and Matthias Raffelsieper}
}

\begin{document}

\newcommand{\SN}{\mbox{\sf SN}}
\newcommand{\real}{\mbox{\bf R}_{\geq 0}}
\newcommand{\ar}{\mbox{\sf ar}}
\newcommand{\argu}{\mbox{\sf arg}}
\newcommand{\ovf}{\mbox{\sf overflow}}
\newcommand{\rt}{\mbox{\sf root}}
\newcommand{\XX}{{\cal X}}
\newcommand{\desda}{\; \Longleftrightarrow \;}
\newcommand{\gehz}{\gtrsim}
\newcommand{\nt}{\mbox{\sf not}}
\newcommand{\ms}{\mbox{\sf morse}}
\newcommand{\inv}{\mbox{\sf inv}}
\newcommand{\tl}{\mbox{\sf tail}}
\newcommand{\tlp}{\mbox{\sf tail0}}
\newcommand{\hd}{\mbox{\sf head}}
\newcommand{\obs}{\mbox{\sf Obs}}
\newcommand{\pob}{\mbox{\sf P}}
\newcommand{\zip}{\mbox{\sf zip}}
\newcommand{\even}{\mbox{\sf even}}
\newcommand{\odd}{\mbox{\sf odd}}
\newcommand{\fib}{\mbox{\sf Fib}}
\newcommand{\ones}{\mbox{\sf ones}}
\newcommand{\zeros}{\mbox{\sf zeros}}
\newcommand{\nat}{\mbox{\bf N}}
\newcommand{\nf}{\mbox{\bf NF}}
\newcommand{\str}{D^{\omega}}
\newcommand{\TT}{{\cal T}_s}
\newcommand{\EE}{{\cal E}}

\newcommand{\outto}
{\mathrel{\smash{\overset{\raisebox{-0.25ex}{\scriptsize o}}{\to}}}}
\newcommand{\noutto}
{\mathrel{\smash{\overset{\raisebox{-.025ex}{\scriptsize no}}{\to}}}}
\newcommand{\nrootto}
{\mathrel{\smash{\overset{\raisebox{-.025ex}{\scriptsize ${>}\epsilon$}}{\to}}}}
\newcommand{\rootto}
{\mathrel{\smash{\overset{\raisebox{-.025ex}{\scriptsize $\epsilon$}}{\to}}}}
\newcommand{\parto}{\overset{\parallel}{\to}}
\newcommand{\noutparto}
    {\mathrel{\smash{\xrightarrow{\raisebox{-.025ex}%
    {\scriptsize $\shortparallel~\text{no}$}}}}}
\newcommand{\Pos}{\mathrm{Pos}}
\newcommand{\Pow}{\mathfrak{P}}
\newcommand{\Cons}{\mathcal{C}}
\newcommand{\fsym}[1]{\mathsf{#1}}
\newcommand{\res}{\mathrm{res}}
\newcommand{\subst}{\varsigma}
\newcommand{\substP}{\varrho}
\newcommand{\hole}{\square}
\newcommand{\oldcomment}[1]{}
\newcommand{\comment}[1]{\footnote{!!! #1}}

\newcommand{\wrt}{w.r.t.\ }
\newcommand{\secref}[1]{\hyperref[#1]{Section~\ref*{#1}}}

% \spnewtheorem*{parallelmoveslemma}{Parallel Moves Lemma}
%     {\bfseries\upshape}{\itshape}

\theoremstyle{plain}
\newtheorem{theorem}{Theorem}
\newtheorem{lemma}[theorem]{Lemma}
\newtheorem{proposition}[theorem]{Proposition}
\newtheorem*{parallelmoveslemma}{Parallel Moves Lemma\label{lem:PML}}
\newcommand{\thmref}[1]{\hyperref[#1]{Theorem~\ref*{#1}}}
\newcommand{\lemref}[1]{\hyperref[#1]{Lemma~\ref*{#1}}}
\newcommand{\propref}[1]{\hyperref[#1]{Proposition~\ref*{#1}}}
\newcommand{\PMLref}{\hyperref[lem:PML]{Parallel Moves Lemma}}

\theoremstyle{definition}
\newtheorem{definition}[theorem]{Definition}
\newtheorem{example}[theorem]{Example}
\newcommand{\defref}[1]{\hyperref[#1]{Definition~\ref*{#1}}}
\newcommand{\exref}[1]{\hyperref[#1]{Example~\ref*{#1}}}

\title{Stream Productivity by Outermost Termination}
\author{Hans Zantema
\institute{
Department of Computer Science, TU Eindhoven, P.O.\ Box 513,\\
5600 MB Eindhoven, The Netherlands\\
\email{H.Zantema@tue.nl}
\\\\
Institute for Computing and Information Sciences, Radboud University
\\
Nijmegen, P.O.\ Box 9010, 6500 GL Nijmegen, The Netherlands
}
\\\\
\and
Matthias Raffelsieper
\institute{
Department of Computer Science, TU Eindhoven, P.O.\ Box 513,\\
5600 MB Eindhoven, The Netherlands\\
\email{M.Raffelsieper@tue.nl}
}
}

\def\authorrunning{Hans Zantema and Matthias Raffelsieper}
\def\titlerunning{Stream Productivity by Outermost Termination}

\maketitle

\begin{abstract}
Streams are infinite sequences over a given data type. A stream
specification is a set of equations intended to define a stream.
A core property is productivity: unfolding the equations produces
the intended stream in the limit. In this paper we show that
productivity is equivalent to termination with respect to the
balanced outermost strategy of a TRS obtained by adding an
additional rule. For specifications not involving branching symbols
balancedness is obtained for free, by which tools for proving
outermost termination can be used to prove productivity fully
automatically. 
\end{abstract}

\section{Introduction}

Streams are among the simplest data types in which the objects are
infinite: they can be seen as maps from the natural numbers to some 
data type $D$.
The basic constructor for streams is the operator `:' mapping a data element $d$
and a stream $s$ to a new stream $d:s$ by putting $d$ in front of $s$.
Using this operator we can define streams by equations. For instance,
the Thue Morse sequence $\ms$ over the data elements $0,1$ can be
specified by the rules
\[ \begin{array}{rclrcl}
\ms & \to & 0:\zip(\inv(\ms),\tl(\ms)) \hspace{8mm} &
\tl(x:\sigma) & \to & \sigma \\
\inv(x:\sigma) & \to & \nt(x) : \inv(\sigma) &
\zip(x:\sigma,\tau) & \to & x : \zip(\tau,\sigma)  \end{array} \]
together with the two rules $\nt(0) \to 1$ and $\nt(1) \to 0$.

This stream specification is {\em productive}: for every $n \in
\nat$ there is a rewrite sequence $\ms \to^* u_1 : u_2 : \cdots : u_n
: t$, that is, by these rules every $n$-th element of the stream can
be computed.
This notion of productivity goes back to Sijtsma~\cite{S89}.
In \cite{EGH08} a nice and powerful approach has been described to prove
productivity automatically for a restricted class of stream
specifications. Here we follow a completely different approach: we do
not have these restrictions, but show that productivity is equivalent
to termination with respect to a particular kind of outermost
rewriting, after adding the rule $x : \sigma \to \ovf$. The intuition
of this equivalence is clear: productivity is 
equivalent to the claim that every ground term rewrites to a term
with ':' on top. This kind of rewriting is forced by doing outermost
rewriting, and as soon as ':' is on top, the reduction to $\ovf$ is
forced, blocking further rewriting. 

However, there are some pitfalls. In the above example the term
$\tl(\ms)$ admits an infinite outermost reduction starting by
\[ \begin{array}{rclrcl}
\tl(\ms) & \to & \tl(0:\zip(\inv(\ms),\tl(\ms))) \\
 & \to & \zip(\inv(\ms),\tl(\ms)) \end{array} \]
and then repeating this reduction forever on the created subterm
$\tl(\ms)$. So the outermost strategy to be considered needs an extra
requirement disallowing this reduction. This requirement is what we
call {\em balanced}: we require every redex in the reduction either
to be reduced eventually, or rewritten by a redex closer to the root.
In the given example the redex $\ms$ in $\zip(\inv(\ms),\cdots)$ is
never reduced, nor rewritten by a higher redex, so the resulting
infinite outermost reduction is not balanced.

Our main result states that a stream specification given by a TRS $R$
is productive for all ground terms if and only if 
$R \cup \{ x : \sigma \to \ovf \}$ does not admit an infinite 
balanced outermost reduction.

For the special case
% not involving defined symbols with arity $> 1$, 
without rewrite rules for the data and without symbols having more than one
argument of stream type,
balancedness is obtained for free, and
productivity of $R$ on all ground terms is equivalent to outermost
termination of $R \cup \{ x : \sigma \to \ovf \}$. For this fully
automatic tools can be used, for
instance based on the approaches of \cite{EH09,RZ09,Thiemann09}. 

As an example consider
\[ \begin{array}{rcl}
\fsym{c} & = & 1:\fsym{c} \\
\fsym{f}(0:\sigma) & = & \fsym{f}(\sigma) \\
\fsym{f}(1:\sigma) & = & 1:\fsym{f}(\sigma) \\
\end{array} \]
by which we want to compute $\fsym{f}(\fsym{c})$. Clearly $c$ only 
consists of ones, and $f$ only removes zeros, so the result of 
$\fsym{f}(\fsym{c})$ will be the infinite stream of ones. Every 1
in this stream is easily produced by the reduction
\[ \fsym{f}(\fsym{c}) \to \fsym{f}(1:\fsym{c}) \to
1: \fsym{f}(\fsym{c}) \to \cdots, \]
proving productivity of $\fsym{f}(\fsym{c})$. 
However, the approach from \cite{EGH08} fails, as this stream 
specification is not \emph{data-obliviously productive}, i.e., the 
identity of the data is essential for productivity. As far as we
know, and confirmed by the authors of \cite{EGH08}, until now there 
were no techniques for proving productivity automatically if the
productivity is not data-oblivious. This has changed by the
approach we present in this paper. The above example does not directly fit
the basic format of our approach. However, it is easily (and
automatically) unfolded to the system $R$ consisting of the rules
\[ \begin{array}{rcl}
\fsym{c} & = & 1:\fsym{c} \\
\fsym{f}(x:\sigma) & = & \fsym{g}(x,\sigma) \\
\fsym{g}(0,\sigma) & = & \fsym{f}(\sigma) \\
\fsym{g}(1,\sigma) & = & 1:\fsym{f}(\sigma) \\
\end{array} \]
fitting the basic format of our approach.
Now outermost termination of $R \cup \{ x : \sigma \to \ovf \}$
can be proved by a tool. Due to the shape of the symbols and the
fact that there are no rewrite rules for the data, also balanced 
outermost termination of $R \cup \{ x : \sigma \to \ovf \}$ can be
concluded. Then the main theorem of our paper states productivity,
not only for $\fsym{f}(\fsym{c})$ but for all ground terms of sort stream.

The approach works for several other examples, for instance for an 
alternative definition of the $\ms$ stream.

In \cite{Z09} a related approach is described, while an implementation of 
that technique is described in \cite{Z09a}. However, there the result is on 
well-definedness of stream specifications, which is a slightly weaker notion than 
productivity. The main result of  \cite{Z09} is that well-definedness of a stream
specification can be concluded from termination of some transformed system: the
observational variant.

\section{The Main Result}

In stream specifications we have two sorts: $s$ (stream) and $d$
(data). We assume the set $D$ of data elements to consist of the 
unique normal forms of 
ground terms over some signature $\Sigma_d$ with respect to some
terminating orthogonal
rewrite system $R_d$ over $\Sigma_d$. Here all symbols of $\Sigma_d$ 
are of type $d^n \to d$ for some $n \geq 0$. 
In the actual stream specification
we have a set $\Sigma_s$ of stream symbols, each being of type
$ d^n \times s^m \to s$ for $n,m \geq 0$. 
Apart from that, we assume a particular symbol ${:} \not\in \Sigma_s$ 
having type $d \times s \to s$. 
As a notational convention 
variables of sort $d$ will be denoted by
$x,y$, terms of sort $d$ by $u,u_i$, 
variables of sort $s$ by $\sigma,\tau$, and
terms of sort $s$ by $t,t_i$.

\begin{definition}
\label{defss}
A {\em stream
specification} $(\Sigma_d,\Sigma_s,R_d,R_s)$ consists of  
$\Sigma_d,\Sigma_s,R_d$ as given before, and a set $R_s$ of rewrite
rules over $\Sigma_d \cup \Sigma_s \cup \{ : \}$ of the shape
\[ f(u_1,\ldots,u_n,t_1, \dots, t_m) \to t, \]
where 
\begin{itemize}
\item $f \in \Sigma_s$ is of type $ d^n \times s^m \to s$,
% \item for every $i = 1,\ldots,n$ the term $u_i$ is either a variable 
% of sort $d$ or $u_i \in D$, 
\item for every $i = 1,\ldots,m$ the term $t_i$ is either a variable 
of sort $s$, or $t_i = x : \sigma$ where $x$ is a variable of sort
$d$ and $\sigma$ is a variable of sort $s$,
\item $t$ is any well-sorted term of sort $s$,
\item $R_s \cup R_d$ is orthogonal,
\item Every term of the shape
$f(u_1,\ldots,u_n,u_{n+1}:t_1,\ldots,u_{n+m}:t_m)$
for $f \in \Sigma_s$ of type $ d^n \times s^m \to s$, and
$u_1,\ldots,u_{n+m} \in D$ matches with the left hand side of a rule from
$R_s$.
\end{itemize}
\end{definition}

Sometimes we call $R_s$ a stream specification: in that case $\Sigma_d$, 
$\Sigma_s$ consist of the symbols of sort $d$, $s$, respectively, occurring in
$R_s$, and $R_d = \emptyset$. Rules $\ell \to r$ in $R_s$ are often written as
$\ell = r$.

\defref{defss} is nearly the same as in \cite{Z09}. It is closely
related to the definition of stream specification in \cite{EGH08}: 
by introducing fresh symbols and rules for defining these fresh symbols,
every stream specification in the format of \cite{EGH08} can be
unfolded to a stream specification in our format. 
In the end of the introduction, where we unfolded 
$\fsym{f}(x:\sigma)$ to $\fsym{g}(x,\sigma)$, we already saw an example of this.

For defining productivity we follow the definition from \cite{EGH08}: a stream
specification is called productive for a ground term $t$ if for every $n \in \nat$ 
there exists a reduction of the shape $t \to^* u_1 : u_2 : \cdots : u_n : t'$.
Instead of fixing the start ground term $t$ we prefer to require 
this for all ground terms of sort $s$. In practice this will make hardly any difference: 
typically a stream specification consists of an intended stream to be defined 
and a few auxiliary functions for which productivity not only holds for the 
single stream to be defined but also for any ground term built from 
it and the auxiliary functions.

Taking all ground terms of sort $s$ instead of only one has a strong advantage: then 
for proving productivity it is sufficient to prove that the first element is produced, 
rather than all elements. This is expressed in the following proposition that will
serve as our characterization of productivity:
\begin{proposition}
\label{prop:Prod}
A stream specification $(\Sigma_d,\Sigma_s,R_d,R_s)$ is productive
for all ground terms of sort $s$ if and only if every ground term $t$ of
sort $s$ admits a reduction $t \to_{R_s \cup R_d}^* u' : t'$.
\end{proposition}

\begin{proof}
The ``only if'' direction of the proposition is obvious. To show the ``if''
direction, we show that if for all ground terms of sort $s$ we have
$t \to_{R_s \cup R_d}^* u' : t'$, then
$t \to_{R_s \cup R_d}^* u_1 : u_2 : \cdots : u_n : t_n$
for all $n \in \nat$. This is done by induction on $n$.

If $n=0$, then the proposition directly holds.

Otherwise, we get from the induction hypothesis that
$t \to_{R_s \cup R_d}^* u_1 : u_2 : \cdots : u_{n-1} : t_{n-1}$.
Since $t_{n-1}$ is also a ground term of sort $s$, we have
$t_{n-1} \to_{R_s \cup R_d}^* u' : t'$ by assumption.
Hence,
$t \to_{R_s \cup R_d}^* u_1 : u_2 : \cdots : u_{n-1} : t_{n-1}
    \to_{R_s \cup R_d}^* u_1 : u_2 : \cdots : u_{n-1} : u' : t'$,
proving the proposition.
\end{proof}

From now on we omit the subscript $R_s \cup R_d$ in rewrite steps
$\to$.
Given a term $t$, we define the set of positions $\Pos(t) \subseteq \nat^*$
as the smallest set such that $\epsilon \in \Pos(t)$
and if $t = f(t_1, \dotsc, t_n)$, then
$i.p' \in \Pos(t)$ for all $1 \le i \le n$ and $p' \in \Pos(t_i)$.
The replacement of the subterm of $t$ at some position $p$, denoted $t|_p$,
by another term $t'$ is denoted $t[t']_p$ and defined by
$t[t']_{\epsilon} = t'$
and
$f(t_1, \dotsc, t_n)[t']_{i.p'} = f(t_1, \dotsc, t_i[t']_{p'}, \dotsc, t_n)$.
A \emph{context} $C$ is a special term, in which the variable $\hole$ occurs
exactly once. Then, we write $C[t]$ to denote the term that is obtained by
replacing $\hole$ with the term $t$.
If in a rewrite step $t \to t'$ the redex is on position $p
\in \Pos(t)$, we write $t \to_p t'$. We also write $t \to_p$ to indicate that
the term $t$ has a redex at position $p$. For two positions $p,q$ we write
$p \leq q$ if $p$ is a prefix of $q$, and $p < q$ if $p$ is a proper
prefix of $q$, that is, the position $p$ is
above $q$.
If neither $p \leq q$ nor $q \leq p$, then we call the two positions
\emph{independent}, which is denoted $p \parallel q$.
A rewrite step $t \to_p t'$ is called 
{\em outermost} if $t$ does not contain a redex in a position $q$
with $q < p$. A reduction is called outermost if every step is
outermost.
Such an infinite outermost reduction is called \emph{balanced outermost},
if every redex is eventually either reduced or consumed by a redex at
a higher position, as formally defined below.

\begin{definition}
\label{def:OutBal}
Let $R$ be an arbitrary TRS.
An infinite outermost reduction 
\[ t_1 \to_{p_1} t_2 \to_{p_2} t_3 \to_{p_3} t_4 \cdots \]
with respect to $R$ is called {\em balanced outermost} if for every $i$ and
every redex of $t_i$ on position $q$ there exists $j \geq i$ such that
$p_j \leq q$.
The TRS $R$ is called \emph{balanced outermost terminating}
if it does not admit an infinite balanced outermost reduction.
\end{definition}

% This definition of balanced outermost reduction can be given for any TRS.
A direct consequence is that for any infinite outermost 
reduction that is not balanced and contains a redex on position $p$ in some 
term, every term later in the reduction has a redex
on position $p$, too.
% It is easily seen that for this property both 
% left-linearity and non-overlappingness are essential.

As an example we consider the stream specification for the Thue Morse 
sequence from the introduction. The infinite reduction
\[ \begin{array}{rclrcl}
\tl(\ms) & \to & \tl(0:\zip(\inv(\ms),\tl(\ms))) \\
& \to & \zip(\inv(\ms),\tl(\ms)) \end{array} \] 
continued by repeating this reduction forever on the created subterm
$\tl(\ms)$, is outermost, but not balanced, since the redex $\ms$ on 
position $1.1$ in the term $\zip(\inv(\ms),\tl(\ms))$ is never rewritten, 
and neither a higher redex. By forcing the infinite outermost reduction 
to be balanced, this redex should be rewritten, after which the rule 
for $\inv$ can be applied, and has to be applied due to balancedness, 
after which the first argument of $\zip$ will have '$:$' as its root, 
after which outermost reduction will choose the $\zip$ rule and create
a '$:$' as the root.

Now we arrive at the main theorem, showing that productivity of a stream
specification is equivalent to balanced outermost termination of the stream
specification extended with the rule $x : \sigma \to \ovf$.

\begin{theorem}
\label{thm:ProdEquivBalOutTerm}
A stream specification $(\Sigma_d,\Sigma_s,R_d,R_s)$ is productive
for all ground terms of sort $s$ if and only if 
\[ R_d \; \cup \; R_s \; \cup \; \{x : \sigma \to \ovf \} \]
is balanced outermost terminating.
\end{theorem}

\section{Soundness}
\label{sec:Soundness}

In this section we show soundness of \thmref{thm:ProdEquivBalOutTerm}, i.e.,
balanced outermost termination of the extended TRS implies productivity of the
corresponding stream specification.

For doing so, using the special shape of stream specifications, first we prove 
a lemma stating that any ground
term not having '$:$' as root symbol contains a redex that is not below a '$:$'
symbol. 

\begin{lemma}
\label{lem:RedNotBelowCons}
Let $(\Sigma_d,\Sigma_s,R_d,R_s)$ be a stream specification, and let $t$ be a
ground term of sort $s$ with $\rt(t) \ne {:}$. Then there exists a position $p
\in \Pos(t)$ such that $t \to_p$ and for all $p' < p$, $\rt(t|_{p'}) \ne {:}$.
\end{lemma}

\begin{proof}
This lemma is proven by structural induction on $t$.

If $t$ is a constant $c \in \Sigma_s$, then by requirement there is a
rule $c \to r \in R_s$ for some term $r$.

Otherwise, $t = f(u_1, \dotsc, u_m, t_1, \dotsc, t_n)$ for some symbol
$f \ne {:}$, ground terms $u_1, \dotsc, u_m$ of sort $d$, and ground
terms $t_1, \dotsc, t_n$ of sort $s$. If $t \to_{\epsilon}$, then the lemma
holds. Therefore, we assume in the rest of the proof that this is not the case.

If there is a $u_i$ such that $u_i \to$, then this reduction is not below a
'$:$' since $f \ne {:}$.

Otherwise, assume that $u_i \in \nf(R_d)$ for all $1 \le i \le m$. If there is a
term $t_j$ with $\rt(t_j) \ne {:}$, then we get from the induction hypothesis
that $t_j \to_p$ for some position $p$ that is not below a '$:$'. Hence, the
position $(m+j).p$ is also not below a '$:$', since $f \ne {:}$. Finally,
we have to consider the case where $u_i \in \nf(R_d)$ and $t_j = u_j : t_j'$ for
all $1 \le j \le n$ and some terms $u_j, t_j'$. However, in this case it is
required by stream specifications that $t \to_{\epsilon}$, giving a
contradiction to our assumption.
\end{proof}

Using the above lemma, we can now prove soundness of our main result, i.e.,
we can show a stream specification $(\Sigma_d, \Sigma_s, R_d, R_s)$ to be
productive by showing $R_d \cup R_s \cup \{ x : \sigma \to \ovf \}$ to be
balanced outermost terminating.

\begin{proof}[Proof of Soundness of \thmref{thm:ProdEquivBalOutTerm}]
Assume $t$ is not productive, i.e., it does not rewrite
to a term with '$:$' as its root symbol. This allows us to construct an infinite
balanced outermost reduction \wrt
$R_d \cup R_s \cup \{ x : \sigma \to \ovf \}$:
According to \lemref{lem:RedNotBelowCons}, there exists a position $p$ such that
$t \to_p$ and for all $p' < p$, $\rt(t|_{p'}) \ne {:}$. Hence, there exists a
position $q_1 \le p$ such that for some term $t_1$, $t \to_{q_1} t_1$ is an
outermost step \wrt $R_d \cup R_s$. Since also for all $q' < q_1$,
$\rt(t|_{q'}) \ne {:}$, this is also an outermost step \wrt $R_d \cup R_s \cup
\{x:\sigma \to \ovf\}$.
Also $t_1$ is not productive, otherwise, if $t_1$ would rewrite to a term with
'$:$' as its root symbol, then so would $t$. Hence, we can repeat this argument
to obtain an infinite outermost reduction $t = t_0 \to_{q_1} t_1 \to_{q_2} t_2
\to_{q_3} \dots$.

There might however be a term $t_i$ and a redex on a position $p \in \Pos(t_i)$
that is never reduced or consumed in the constructed infinite outermost
reduction. However, then there is never a reduction step above $p$ in the
remaining reduction, i.e., for all $j > i$, $q_j \not\le p$.
Since the reduction consists of outermost steps, we furthermore can conclude
that $q_j \not> p$, otherwise $t_{j-1} \to_{q_j} t_j$ would not be outermost.
Hence, $q_j \parallel p$ for all $j > i$. Let $p' \le p$ such that $t_i
\to_{p'}$ is an outermost step. Then also $p' \parallel q_j$ for all $j > i$,
since $q_j \le p' \le p$ would contradict the assumption that $q_j \not\le p$
and $q_j > p'$ would contradict the assumption that $t_{j-1} \to_{q_j} t_j$ is
an outermost step. Therefore, we can reduce the redex at position $p'$ at any
time, without affecting reducibility of the redexes at positions $q_j$. These
however might now become non-outermost steps. So let $t_0 \to^* t_i
\to_{q_{i+1}} \dots \to_{q_k} t_k \to_{p'} t_{k+1}'$ for some $k > i$ such that
$t_{k+1}' \to_{q_{k+1}}$ is not an outermost step. But then we can again apply
the above reasoning that there is a redex on a position not below a '$:$' symbol
in $t_{k+1}'$ and following terms, yielding another infinite outermost
reduction for which the redex of $t_i$ at position $p$ is reduced or consumed.
Repeating this construction gives an infinite balanced outermost reduction,
which shows soundness of the theorem.
\end{proof}

\section{Completeness}
\label{sec:Completeness}

In this section we show completeness of \thmref{thm:ProdEquivBalOutTerm}, i.e.,
disproving balanced outermost termination allows us to conclude
non-productivity.
Before we can prove this however, we first have to introduce some notation that
allows us to distinguish between outermost and non-outermost rewrite steps.

\begin{definition}
\label{def:NonOutermost}

For a TRS $R$, we define $t \outto_p t'$ if $t \to_p t'$ is an outermost
rewrite step.
Otherwise, if $t \to_p t'$ is not an outermost rewrite step, we define
$t \noutto_p t'$.
\end{definition}

By convention, we will denote substitutions with $\subst, \substP$, which are
mappings from variables to terms, written as $\{ x_1 := t_1, \dotsc, x_n := t_n
\}$. Application of a substitution $\subst$ to a term $t$ is denoted $t\subst$.
Given a TRS $R$, $c$ is called a \emph{constructor} if $\rt(\ell) \ne c$
for all rules $\ell \to r \in R$. 
Furthermore, given a term $t$, the \emph{tail} of a position
$p \in \Pos(t)$ \wrt another position $p' \in \Pos(t)$ with $p' \le p$ is
denoted $p \smallsetminus p'$ and defined as $p \smallsetminus \epsilon = p$
and $i.p \smallsetminus i.p' = p \smallsetminus p'$.
Thereby, $p \smallsetminus p'$ is $p$ after removing the prefix $p'$.
Finally, we define the concept of parallel rewrite steps.

\begin{definition}
\label{def:Parallel}

For a TRS $R$ we define the \emph{parallel rewrite step} $t \parto t'$ if there
exists a set of positions $\{ p_1, \dotsc, p_n \} \subseteq \Pos(t)$
such that for all $1 \le i,j \le n$ with $i \ne j$, $p_i \parallel p_j$ and
$t \to_{p_1} t_1 \to_{p_2} \dotsb \to_{p_n} t'$.
\end{definition}

A standard lemma that we will use is the \emph{Parallel Moves Lemma},
which is for example presented and proved in \cite[Lemma~6.4.4]{BaaderNipkow98}.
We will however use a slightly different form than presented there, but the
proof of~\cite{BaaderNipkow98} easily shows this to be true.

\begin{parallelmoveslemma}
% \label{lem:PML} % The label is already provided
Let $R$ be a TRS and $\ell \to r \in R$ a left-linear rule. If for two
substitutions $\subst, \subst'$ we have that $x\subst \parto x\subst'$ for
all variables $x$, then $\ell\subst \parto \ell\subst' \to r\subst'$ and
$\ell\subst \to r\subst \parto r\subst'$.
\end{parallelmoveslemma}

It is easy to see that for an orthogonal TRS, the \PMLref{} is always applicable
in case a term is reducible at two different positions. This holds,
since there are no overlaps of the rules, i.e., any redex contained in
another redex must be below some variable position, hence in the substitution
part.

We will now show that a non-outermost reduction step followed by an outermost
reduction step is either on an independent position or on a position below the
outermost step.

\begin{lemma}
\label{lem:NOindep}
Let $R$ be an orthogonal TRS.
If $t_1 \noutto_{q} t_2 \outto_{p} t_3$, then $p \parallel q$ or $p < q$.
\end{lemma}

\begin{proof}
Let $t_1 \noutto_{\ell_1 \to r_1,q} t_2 \outto_{\ell_2 \to r_2, p} t_3$.
Therefore, a position $q' < q$ exists such that $t_1 \to_{\ell' \to r', q'}$.

Assume that $p \nparallel q$ and $q \le p$ (where the latter implies the
former). Then
$t_1 = t_1[\ell'\subst'[\ell_1\subst_1]_{q \smallsetminus q'}]_{q'}$.
Since $R$ is orthogonal, there exists a variable $x$ and a context $C$ such that
$t_1 = t_1[\ell'\subst''\{x := C[\ell_1\subst_1]\}]_{q'}$, where $\subst''$ is
like $\subst'$ except that $\subst''(x) = x$. Therefore,
$t_1 = t_1[\ell'\subst''\{x := C[\ell_1\subst_1]\}]_{q'} \noutto_{q}
t_1[\ell'\subst''\{x := C[r_1\subst_1]\}]_{q'} = t_2$.
In this last term, the redex at position $p$ is contained, i.e.,
$t_2 = t_1[\ell'\subst''\{x := C[r_1\subst_1]\}]_{q'} =
t_1[\ell'\subst''\{x := C[r_1\subst_1[\ell_2\subst_2]_{p'}]\}]_{q'}$
for a position $p'$ such that $p = q' . (q \smallsetminus q') . p'$.
However, this contradicts our assumption $t_2 \outto_p t_3$, since
$t_2 \to_{q'}$ and $q' < q \le p$.
\end{proof}

The above lemma allows us to show that for such a sequence of steps, i.e., a
non-outermost step followed by an outermost step, we can swap the evaluation
order and still reach the same term. In the remainder of this section we denote
with $\noutparto_P$ parallel non-outermost steps, i.e., a parallel reduction
where all positions in the set $P$ are on non-outermost positions.

\begin{lemma}
\label{lem:SwapNOwithO}
For an orthogonal TRS $R$, if $t_1 \noutparto t_2 \outto_p t_3$,
then $t_1 \outto_p t \outto^* t' \noutparto t_3$ for some terms $t, t'$.
\end{lemma}

\begin{proof}
Let $t_1 \noutparto_Q t_2 \outto_p t_3$ for some $Q \subseteq \Pos(t_1)$.
By \lemref{lem:NOindep}, we get that either $q \parallel p$ or $q > p$ for all
$q \in Q$.

If $q \parallel p$ for all $q \in Q$, then we can swap the two reductions, i.e.,
$t_1 \outto_p t \to_q t_3$ for some term $t$.
If $t \noutto_q t_3$, then we have the required shape.
Otherwise, if $t \outto_q t_3$ then we also have the required shape, since
$t_1 \outto_p t \outto_q t_3 \noutparto_{\emptyset} t_3$.
% We will now show that $t \noutto_q t_3$.
% Since $t_1 \noutto_q t_2$, there exists a position $q' < q$ such that $t_1
% \to_{q'}$. If $p \le q'$, then $p \le q' < q$ which contradicts our assumption
% $p \parallel q$. If however $p > q'$, then due to orthogonality of $R$ we have
% $t_2 \to_{q'}$ and thus the contradiction $t_2 \not\outto_{p}$. Hence, $p
% \parallel q'$, therefore $t \to_{q'}$ and therefore $t \noutto_q t_3$.

Otherwise, a maximal $\emptyset \ne Q' \subseteq Q$ exists such that $p < q'$
for all $q' \in Q'$.
Let $t_1 \noutparto_{Q'} t_2' \outto_p t_3'$.
Then, since $R$ is orthogonal, we can apply the
\PMLref{}, showing that $t_1 \outto_p t \parto t_3'$ for some $t$. All
redexes in the reduction $t \parto t_3'$ are on independent positions, hence
we can first reduce all outermost ones, then all non-outermost ones. Therefore,
a term $t'$ exists such that $t_1 \outto_p t \outto^* t' \noutparto_{Q''} t_3'$
for some set $Q''$, where $p < q''$ for all $q'' \in Q''$.
Because all positions in $Q \setminus Q'$ are independent from the
position $p$, they are also independent from the positions in $Q''$. Thus, we
get that $t_1 \outto_p t \outto^* t' \noutparto_{Q'' \cup (Q \setminus Q')}
t_3$.
\end{proof}

Using the above lemma, we can prove that any reduction can be split into an
outermost and a non-outermost reduction.

\begin{lemma}
\label{lem:SplitOandNO}
Let $R$ be an orthogonal TRS.

If $t \to^* t'$, then $t \outto^* \hat{t} \noutto^* t'$ for some $\hat{t}$.
\end{lemma}

\begin{proof}
Let $t \to^n t'$. We perform induction on the length $n$ of this reduction.

If $n=0$, then $t=t'$ and nothing has to be shown.

Otherwise, let $t \to^{n-1} t_{n-1} \to t'$. We get from the induction
hypothesis that $t \outto^* \hat{t}' \noutto^* t_{n-1} \to t'$ for some
$\hat{t}'$. If $t_{n-1} \noutto t'$ then the lemma holds. So assume $t_{n-1}
\outto t'$. Then $\hat{t}' \noutto^* t_{n-1} \outto t'$ and therefore
$\hat{t}' \noutparto^* t_{n-1} \outto t'$.
Repeated application of \lemref{lem:SwapNOwithO} shows that for some $\hat{t}$,
$\hat{t}' \outto^* \hat{t} \noutparto^* t'$, hence
$t \outto^* \hat{t}' \outto^* \hat{t} \noutto^* t'$
by unfolding the parallel non-outermost steps, which proves the lemma.
\end{proof}

This allows us to show that for checking the productivity criterion of
\propref{prop:Prod}, we only have to consider outermost reductions.

\begin{lemma}
\label{lem:OutermostProd}
Let $R$ be an orthogonal TRS having a binary symbol $:$ in its signature.

If $t \to^* u : t_u$, then $t \outto^* u' : t_u' \to^* u : t_u$ for some terms
$u'$, $t_u'$.
\end{lemma}

\begin{proof}
Let $t \to^* u:t_u$. Then by \lemref{lem:SplitOandNO}, $t \outto^* \hat{t}
\noutto^* u : t_u$. If $\rt(\hat{t}) = {:}$, then the lemma holds.
Otherwise, $\rt(\hat{t}) \ne {:}$.
Let $\hat{t} = t_0 \noutto t_1 \noutto \dotsb \noutto t_k = u : t_u$.
Then for all $1 \le i \le k$, $\rt(t_i) = \rt(t_{i-1})$, since none of the terms
can be reduced at the root position as this would be an outermost reduction
step. This however gives a contradiction, because ${:} \ne \rt(\hat{t}) =
\rt(t_0) = \rt(t_1) = \dots = \rt(t_k) = \rt(u:t_u) = {:}$.
\end{proof}

Next, we prove two technical lemmas that will be used to prove completeness of
our main theorem.
In the first we handle the case where a redex in a term that starts an infinite
balanced outermost reduction is also reduced at that position later in the
infinite balanced outermost reduction. In this case, we can bring forward this
step and still get an infinite balanced outermost reduction.

\begin{lemma}
\label{lem:BalOutStaysBalOutOnEqPos}
Let $R$ be an orthogonal TRS for which '$:$' is a constructor.

If $t_0 \outto_{p_1} t_1 \outto_{p_2} \dotsb \outto_{p_j} t_j \outto_{p_{j+1}}
\dotsb$ is an infinite balanced outermost reduction, where for all $i \in
\nat$, $\rt(t_i) \ne {:}$, $t_0 \outto_{p_j} t_1'$, and $p_i \not\le p_j$ for
all $1 \le i < j$, then an infinite balanced outermost reduction
$t_0 \outto_{p_j} t_1' \outto \dotsb$ exists,
where for all $i \in \nat$, $\rt(t_i') \ne {:}$.
\end{lemma}

\begin{proof}
First we show that $p_i \parallel p_j$ for all $1 \le i < j$.
For this, we perform induction on $j-i$ and prove that if $i<j$ and $t_{i-1}
\to_{p_j}$, then $p_i \parallel p_j$ and $t_i \to_{p_j}$.

If $j-i=0$, then $i = j$ and the claim vacuously holds.
Otherwise, we have $t_{i-1} \outto_{p_i} t_i$ and $t_{i-1} \to_{p_j}$. If $p_i
\le p_j$, we have a contradiction to the requirement $p_i \not\le p_j$, since $i
< j$. If $p_i > p_j$, then we also have a contradiction, since then $t_{i-1}
\not\outto_{p_i}$. Hence, $p_i \parallel p_j$ and therefore also $t_i
\to_{p_j}$.

This shows that all positions $p_i$ with $1 \le i < j$ are on independent
positions from $p_j$, since $t_0 \to_{p_j}$ by assumption. Therefore, we can
swap their order and get a reduction $t_0 \outto_{p_j} t_1' \to_{p_1} t_2'
\to_{p_2} \dots \to_{p_{j-1}} t_j' = t_j$.
Due to \lemref{lem:SplitOandNO}, there exists a $\hat{t}_1$ such that
$t_0 \outto_{p_j} t_1' \outto^* \hat{t}_1 \noutto^* t_j$.
Let $t_0 = t_0' \outto_{q_1} t_1' \outto_{q_2}
\dotsb \outto_{q_k} t_k' = \hat{t}_1 \noutto_{q_{k+1}} \dotsb \noutto_{q_l} t_l'
= t_j$, where $q_1 = p_j$. Furthermore, let $q_{l+m} = p_{j+m}$ and
$t_{l+m}' = t_{j+m}$ for all $m \ge 1$.
We will now show that every redex in this reduction is eventually reduced or
consumed by a higher redex.

Assume not, i.e., there exists $i \ge 0$ such that for some $q \in \Pos(t_i')$,
$t_i' \to_q$ and for all $m>i$, $q_m \not\le q$, i.e., either $q_m > q$ or $q_m
\parallel q$. We can conclude that $i < l$, since $t_l' = t_j$ and $t_j$ is part
of the balanced outermost reduction $t_0 \outto_{p_1} \dotsb$. If $0 \le i < k$,
then $t_i' \outto_{q_{i+1}} t_{i+1}' \outto_{q_{i+2}} \dotsb \outto_{q_k} t_k'$.
If $q_{i+1} > q$, then because of $t_i \to_q$ we would have
$t_i \not\outto_{q_{i+1}}$; therefore this cannot occur.
If $q_{i+1} \parallel q$, then we also have
$t_{i+1}' \to_{q}$. Applying this repeatedly shows that $t_k' \to_q$, i.e.,
it suffices to investigate the case where $i \ge k$. In this case, we have $t_i'
\noutto_{q_{i+1}} \dotsb \noutto_{q_l} t_l' = t_j$. If $q_{i+1} \parallel q$,
then also $t_{i+1}' \to_q$. Otherwise, if $q_{i+1} > q$, then due to the
\PMLref{}, we also have $t_{i+1}' \to_q$. Applying this repeatedly
shows that $t_l' = t_j \to_q$ and for all $m > l$ we have $q_m \not\le q$. This
however is a contradiction, since $t_j$ was contained in the initial balanced
outermost reduction. This shows our claim.

Furthermore, any non-outermost step of the above reduction, i.e., any step
$t_{i}' \noutto_{\ell_{i+1} \to r_{i+1}, q_{i+1}} t_{i+1}'$ for $k \le i < l$,
is below some position $p_m$ for $m > j$. To show this, let $t_i' = t_i'
[\ell_{i+1}\subst_{i+1}]_{q_{i+1}}$. Then a position $q' < q_{i+1}$ exists such
that $t_i' = t_i' [\ell\subst[\ell_{i+1}\subst_{i+1}]_{q_{i+1} \smallsetminus
q'}]_{q'} \to_{q'}$ for some $\ell \to r \in R$. Since $R$ is orthogonal, there
must be a variable $x$ and a context $C$ such that $t_i' = t_i' [\ell\subst'\{ x
:= C[\ell_{i+1}\subst_{i+1}] \}]_{q'}$, where $\subst'$ is like $\subst$, except
that $\subst'(x) = x$. Then $t_i' = t_i' [\ell\subst'\{ x :=
C[\ell_{i+1}\subst_{i+1}] \}]_{q'} \noutto_{q_{i+1}} t_i' [\ell\subst'\{ x :=
C[r_{i+1}\subst_{i+1}] \}]_{q'} = t_{i+1}' \to_{q'}$, i.e., $t_{i+1}'$ still
contains a redex at position $q'$.
Repeating this argument, we see that for every reduced non-outermost redex,
there is a still a redex above it in the term $t_l' = t_j$. However, for every
such redex at some position $q'$, there is a position $p_m$ with $m > j$ such
that $p_m \le q'$ due to the initial balanced outermost reduction, showing our
claim.

To the reduction $t_0 = t_0' \outto_{p_j} t_1' \outto^* \hat{t}_1 \noutto^* t_j
\outto_{p_{j+1}} \dotsb$ we can now repeatedly apply \lemref{lem:SwapNOwithO}
to get the outermost reduction $t_0 = t_0' \outto_{p_j} t_1' = t_1'' \outto^*
\hat{t}_1 \outto_{p_{j+1}} t_2'' \outto^* \hat{t}_2 \outto_{p_{j+2}} \dotsb$.
This is a balanced outermost reduction due to the above observations, since
every redex in a reduction $t_i'' \outto^* \hat{t}_i$ is eventually reduced or
consumed and every redex in a reduction $\hat{t}_i \noutto^* t_{i+1}''$ is below
some position $p_m$ that is reduced later in the reduction.

Finally, we have to show that none of the terms in the constructed infinite
balanced outermost reduction has a '$:$' symbol as its root. If this was not the
case, there would be a term $t''$ with $\rt(t'') = {:}$ and $t'' \outto^*
\hat{t}_m$ for some $m$. However, for every such term $\hat{t}_m$, we have that
$\hat{t}_m \noutto^* t_n$ for some $n$. Since '$:$' is a constructor of $R$, we
would have that ${:} = \rt(t'') = \rt(\hat{t}_m) = \rt(t_n) \ne {:}$, giving a
contradiction and hence showing the desired property.
\end{proof}

The second case we have to consider is that a redex in a term starting an
infinite balanced outermost reduction is strictly below some reduction step. But
also in this case, we will show that we can reduce the redex and still get an
infinite balanced outermost reduction.

\begin{lemma}
\label{lem:ProdBelowBalOut}
Let $R$ be an orthogonal TRS for which ${:}$ is a constructor,
$t_0 \outto_{p_1} t_1 \outto_{p_2} \dots$ be an infinite balanced outermost
reduction with $\rt(t_i) \ne {:}$ for all $i \ge 0$,
$t_0 \to_{\ell \to r, q_1} t_0[r\subst]_{q_1} = t_1'$, and let $p_j \le q_1$
be minimal, with $p_j < q_1$.

Then an infinite balanced outermost reduction $t_0[r\subst]_{q_1} = t_1'
\outto_{p_1'} t_2' \outto_{p_2'} \dots$ exists with $\rt(t_i') \ne {:}$ for all
$i \ge 1$.
\end{lemma}

\begin{proof}
Let $t_{j-1} = t_0 [r_1\subst_1]_{p_1} \dotsc [r_{j-1}\subst_{j-1}]_{p_{j-1}}
[\ell_j\subst_j]_{p_j}$.
Then for some variable $x \in V(\ell_j)$ and some context $C$ we have that
$t_{j-1} = t_0 [r_1\subst_1]_{p_1} \dotsc [r_{j-1}\subst_{j-1}]_{p_{j-1}}
[\ell_j\subst_j' \{ x := C[\ell\subst] \} ]_{p_j} \outto_{p_j} t_0
[r_1\subst_1]_{p_1} \dotsc [r_{j-1}\subst_{j-1}]_{p_{j-1}} \linebreak[5]
[r_j\subst_j' \{ x := C[\ell\subst] \} ]_{p_j} = t_j$,
where $\subst'_j$ is like $\subst_j$, except that $\subst_j'(x) = x$.
% 
% Therefore, we have the reduction
% $t_0[r\subst]_{q_1} \outto_{p_1} \dots \outto_{p_j}
% t_0 [r_1\subst_1]_{p_1} \dotsc\linebreak[5] [r_{j-1}\subst_{j-1}]_{p_{j-1}}
% [r_j\subst_j' \{ x := C[r\subst] \}]_{p_j} = \hat{t}_j$.

If $x \notin V(r_j)$, then the lemma trivially holds.

Otherwise, let $p_{\hole} \in \Pos(C)$ such that $C|_{p_{\hole}} = \hole$ and
$Q_j = \{ q \in \Pos(t_j) ~|~ q = p_j.p'.p_{\hole} \text{ and }
r_j|_{p'} = x \} = \{ q_1^j, \dotsc, q_{m_j}^j \}$.
Then $t_j \to_{\ell \to r, q}$ for all $q \in Q_j$.
Furthermore, we define for all $k > j$, $Q_{k} = \{q_1^k, \dotsc, q_{m_k}^k\} =
(Q_{k-1} \setminus \{ q \in Q_{k-1} ~|~ p_k \le q \})
\cup
\{ q' \in \Pos(t_k) ~|~ \exists q \in Q_{k-1}: q = p_k . p . p',
\ell_k|_{p} = y \in V,
\text{ and } q' = p_k . p'' . p' \linebreak
\text{where } r_k|_{p''} = y
\}$, i.e., we update the set of positions such that independent positions are
kept, positions that are reduced are removed, and positions below a reduction of
the infinite balanced outermost reduction are modified such that they reflect
the position of the redex in the right-hand side. This can be
done since the TRS is orthogonal, which especially implies that a contained
redex cannot overlap with the left-hand side of a rule that is applied above it,
therefore it has to be below a variable position in the left-hand side.

Hence, we have for all $k > j$ either $p_k \parallel q$ for all $q \in Q_{k-1}$,
$p_k = q$ for some $q \in Q_{k-1}$, or $p_k < q$ for some $q \in Q_{k-1}$
($p_k$ cannot be below some $q$, since otherwise it would not be outermost).
In the first case, the reduction
$t_{k-1} [r\subst]_{q_1^{k-1}} \dotsc [r\subst]_{q_m^{k-1}} \to_{p_k}$
is unaffected. In the second case, where $p_k = q$, we can remove this reduction
step. In the third and final case, where $p_k < q$, this reduction is also still
possible, since $R$ is orthogonal and reductions inside another redex cannot
destroy the outer redex. Hence, we can again apply the argument and get an
infinite reduction $t_1' = t_0[r\subst]_{q_1} \to_{p_1'} t_2' \to_{p_2'} \dots$,
where the positions $p_i'$ are the positions $p_i$ after removing reduction
steps as described above. This reduction is balanced, but not
necessarily outermost. However, we can repeatedly apply
\lemref{lem:SwapNOwithO} to get an infinite outermost reduction, which will
defer non-outermost steps forever. To see that this reduction is balanced,
assume the contrary. Then, a term $t_a'$
and a position $q \in \Pos(t_a')$ exist such that
$t_a'|_{q} \to$ and this redex is never reduced or consumed, and there exists
$h > a$ such that $p_h \le q$ since the non-outermost reduction was balanced.
Since \lemref{lem:SwapNOwithO} only swaps non-outermost reductions to the end,
it must be the case that all $p_h \le q$ are non-outermost. Then however an
outermost position $p_{h'} < p_h$ exists, hence it is not deferred
forever. This gives a contradiction, since this position is reduced eventually,
consuming the redex at position $q$.

Finally, we show that $\rt(t_i') \ne {:}$ for all $i \ge 1$. Assume this
not to be the case, i.e., there is a minimal $t_i'$ with $\rt(t_i') = {:}$. Then
$p_i' = \epsilon$ and for $t_{i-1}' \to_{\ell' \to r', p_i'} u : t'$ it must be
the case that $\rt(r') = {:}$. However, since this step was also contained in
the original infinite balanced outermost reduction, this would contradict the
requirement that $\rt(t_i) \ne {:}$. Furthermore, since ${:}$ is a constructor,
also reordering the reductions into an outermost reduction cannot introduce a
term with ${:}$ as root symbol, since otherwise this term could be reduced to a
term $t_i'$ with $\rt(t_i') = {:}$, which we have shown to be false. This proves
the lemma.
\end{proof}

Using the above lemmas, we can finally prove completeness of our main theorem.

\begin{proof}[Proof of Completeness of \thmref{thm:ProdEquivBalOutTerm}]
Assume $R_d \cup R_s \cup \{ x:\sigma \to \ovf \}$
is not balanced outermost terminating, but $(\Sigma_d, \Sigma_s, R_d, R_s)$ is
productive. Then a term $t$ exists that allows an infinite balanced
outermost reduction $t = t_0' \outto t_1' \outto t_2' \outto \dotsb$ and there
exists a reduction $t \to^* u' : t_u'$. Since the symbol $\ovf$ does not occur
on any left-hand side of $R_d \cup R_s$, we conclude that for all $i \ge 0$,
$\rt(t_i') \ne {:}$, since otherwise the rule $x:\sigma \to \ovf$ would be
applicable and no further reductions would be possible.

We can also construct an infinite balanced outermost reduction \wrt $R_d \cup
R_s$ from the given one by removing all applications of the rule $x:\sigma \to
\ovf$, since the symbol $\ovf$ does not occur on any left-hand side of $R_d \cup
R_s$. This might leave some redexes that previously were contained in a redex
\wrt that rule. However, these redexes can only be on positions above which
never a reduction step takes place, hence we can reduce them at any time. Thus,
we have an infinite balanced outermost reduction $t = t_0 \outto_{p_1} t_1
\outto_{p_2} t_2 \outto_{p_3} \dotsb$ \wrt the orthogonal TRS $R_d \cup R_s$,
where for all $i \ge 0$, $\rt(t_i) \ne {:}$.

By \lemref{lem:OutermostProd} we get that an outermost reduction $t
\outto_{q_1} t_1^p \outto_{q_2} \dotsb \outto_{q_n} t_n^p = u : t_u$ exists. Due
to the definition of balanced outermost reductions, we have that a minimal $j$
exists such that $p_j \le q_1$.
Case distinction on the relation of $p_j$ and $q_1$ is performed. If $p_j = q_1$
then we get from \lemref{lem:BalOutStaysBalOutOnEqPos} an infinite balanced
outermost reduction $t_0 \outto_{p_j} t_1^p = t_1' \outto t_2' \outto \dots$.
Otherwise, if $p_j < q_1$, \lemref{lem:ProdBelowBalOut} gives us an
infinite balanced outermost reduction $t_1^p = t_1' \outto t_2' \outto
\dots$.

In both cases, we furthermore have that $\rt(t_i') \ne {:}$ for all $i > 0$.
Hence, by induction on $n$ we get an infinite balanced outermost reduction $u :
t_u \outto \dots$ in which no term has as root symbol '$:$', which yields the
desired contradiction and therefore completes the proof.
\end{proof}

\section{Using Outermost Termination Tools}
\label{sec:BalancednessForFree}

As stated in the introduction, balancedness is obtained for free in case there
are no rewrite rules for the data, i.e., $R_d = \emptyset$, and there are no
rules in $R_s$ that have more than one argument of stream type $s$. In this
section we prove that claim, which allows us to apply automatic tools for
proving outermost termination to show productivity of stream specifications.

\begin{proposition}
\label{prop:BalFree}
Let $(\Sigma_d, \Sigma_s, R_d, R_s)$ be a stream specification with
$R_d = \emptyset$ and the type of all $f \in \Sigma_s$ is of the form
$d^n \times s^m \to s$ for some $n \in \nat$, $m \in \{ 0,1 \}$.

Then every infinite outermost reduction $t_0 \outto t_1 \outto t_2 \outto \dots$
is balanced.
\end{proposition}

\begin{proof}
We perform structural induction to show that for any reduction step
$t \outto_p t'$, we have that $p \le p'$ for all positions $p' \in \Pos(t)$ with
$t \to_{p'}$.

If $t = c \in \Sigma_s$, then by requirement of stream specifications we have
that $t \to_{\epsilon}$, hence $p = \epsilon$. Since $\epsilon \le p'$ for all
$p' \in \Pos(t)$, we haven proven this case.

Otherwise, if $t = f(u_1, \dotsc, u_n)$ (i.e., there is no argument of stream
type), then we again conclude that $t \to_{\epsilon}$. This is due to $R_d = \emptyset$
and the requirements of stream specifications, note that no data operations are 
allowed with arguments of stream type. So we have also proven this case.

In the final case to consider, we have $t = f(u_1, \dotsc, u_n, t')$.
If $t \to_{\epsilon}$, then again we must have that $p = \epsilon$ and hence
have proven the case. Therefore, assume that $p > \epsilon$. Since
$u_1, \dotsc, u_n \in D = \nf(R_d)$, because $R_d = \emptyset$, it must be the
case that for all $p' \in \Pos(t)$ with $t \to_{p'}$, $n+1 \le p'$, hence this
especially holds for $p$ as well. Therefore, we get from the induction
hypothesis that for the reduction step $t' \outto_{p \smallsetminus n+1} t''$,
$p \smallsetminus n+1 \le p''$ for all positions $p'' \in \Pos(t')$ with
$t' \to_{p''}$. Because $t \outto_p t'$, $p = (n+1) . (p \smallsetminus n+1)$,
and for all $p' \in \Pos(t)$ with $t \to_{p'}$ we have $p' = (n+1) . p''$,
it also holds that $p \le p'$, proving this final case and therefore the
proposition.
\end{proof}

The specification of the Thue Morse sequence given in the introduction shows
the necessity of requiring at most one argument to be of stream type.
It was already observed that the infinite reduction
\[\tl(\ms) \to \tl(0 : \zip(\inv(\ms),\tl(\ms))) \to \zip(\inv(\ms),\tl(\ms))
\to \dots,\]
continued by repeatedly reducing the redex $\tl(\ms)$, is outermost
but not balanced.
To show that also the requirement $R_d = \emptyset$ is needed, we again give
an example that allows to construct an infinite outermost reduction that is not
balanced. Consider the stream specification

\[
\begin{array}{rcl}
\tl(x:\sigma) &=& \sigma
\\
\fsym{c} &=& 0 : \fsym{f}(\nt(1),\tl(\fsym{c}))
\\
\fsym{f}(0,\sigma) &=& 1 : \fsym{f}(0,\sigma)
\\
\fsym{f}(1,\sigma) &=& 0 : \fsym{f}(1,\sigma)
\end{array}
\]
together with the rules $R_d = \{ \nt(0) \to 1, \nt(1) \to 0 \}$. This stream
specification is productive, as can be checked with the productivity tool
of~\cite{EGH08}. However, there also exists an infinite outermost reduction,
namely
\[
\tl(\fsym{c}) \to \tl(0 : \fsym{f}(\nt(1),\tl(\fsym{c})))
    \to \fsym{f}(\nt(1),\tl(\fsym{c}))
    \to \dots,
\]
which is continued by repeatedly reducing the redex $\tl(\fsym{c})$. This
redex is outermost, since both rules having the symbol $\fsym{f}$ as root
require either $0$ or $1$ as first argument. To apply one of these rules,
the outermost redex $\nt(1)$ would have to be reduced first, which shows
that the above infinite outermost reduction is not balanced.

To also present an example that does satisfy the requirements of
\propref{prop:BalFree}, we give an alternative definition of the Thue Morse
stream presented in the introduction:
\[
\begin{array}{rcl}
\ms &=& 0 : \fsym{c}
\\
\fsym{c} &=& 1:\fsym{f}(\fsym{c})
\\
\fsym{f}(0:\sigma) &=& 0:1:\fsym{f}(\sigma)
\\
\fsym{f}(1:\sigma) &=& 1:0:\fsym{f}(\sigma)
\end{array}
\]
This example does not fit our format of stream specifications, however unfolding
it leads to a stream specification that still satisfies the requirements of
\propref{prop:BalFree}.
After adding the rule $x : \sigma \to \ovf$, we have to show
outermost termination of the following TRS:
\[
\begin{array}{rcl}
\ms &\to& 0 : \fsym{c}
\\
\fsym{c} &\to& 1:\fsym{f}(\fsym{c})
\\
\fsym{f}(x:\sigma) &\to& \fsym{g}(x,\sigma)
\\
\fsym{g}(0,\sigma) &\to& 0:1:\fsym{f}(\sigma)
\\
\fsym{g}(1,\sigma) &\to& 1:0:\fsym{f}(\sigma)
\\
x:\sigma &\to& \ovf
\end{array}
\]
Outermost termination of the above TRS can for instance be proven using the
transformation of~\cite{RZ09} and AProVE~\cite{AProVE06} as a termination
prover, or using the approach presented in~\cite{EH09}.
This allows to conclude that the above stream specification is productive.

The next example is interesting, since it is not \emph{friendly nesting}, a
condition required by~\cite{EGH08} to be applicable. Essentially, a stream
specification is friendly nesting if the right-hand sides of every nested symbol
start with '$:$', which is clearly not the case for the second rule below.
\[
\begin{array}{rcl}
\fsym{c} &=& 1:\fsym{c}
\\
\fsym{f}(x:\sigma) &=& \fsym{g}(x,\sigma)
\\
\fsym{g}(0,\sigma) &=& 1:\fsym{f}(\sigma)
\\
\fsym{g}(1,\sigma) &=& 0:\fsym{f}(\fsym{f}(\sigma)))
\end{array}
\]
As it can be checked, the above example fits into the stream specification
format considered in this paper and it satisfies the requirements of
\propref{prop:BalFree}. After adding the rule $x : \sigma \to \ovf$, outermost
termination can be proved automatically using the above techniques, which allows
to conclude productivity of the example.

\section{Conclusions}
\label{sec:Conclusion}

We have shown that productivity of a stream specification $(\Sigma_d, \Sigma_s,
R_d, R_s)$ is equivalent to showing outermost balanced termination of $R_d \cup 
R_s \cup \{ x:\sigma \to \ovf \}$. To the best of our knowledge, this is the
first approach capable of proving productivity of stream specifications that are
not data-obliviously productive. It turns out that soundness of this technique
for proving productivity coincides with the easier direction of our equivalence:
outermost termination of the extended TRS implies productivity.

Our format of stream
specifications is more restrictive than the format of~\cite{EGH08}. However, this 
is not an essential restriction as any stream
specification in the latter format can be transformed into our format by
introducing new rules, as illustrated in~\cite{Z09} and at the end of the
introduction of this paper.

It seems that productivity has some relationship with top termination of the stream
specification. However, these notions are not equivalent. For instance, consider
the stream specification
\[
\begin{array}{rcl}
\fsym{c} &=& \fsym{f}(\fsym{c})
\\
\fsym{f}(x:\sigma) &=& \fsym{c}
\end{array}
\]
One easily shows that this system is top terminating, but $\fsym{c}$ is not
productive. We do not see how proving top termination can help for proving
productivity.

When restricting to stream specifications with $R_d = \emptyset$ and where every
left-hand side of $R_s$ contains at most one argument of type $s$, then 
balancedness is obtained for free and techniques for proving outermost
termination can be used to show productivity. An immediate topic for future work 
is hence to devise techniques for proving balanced outermost termination, which
would allow to show productivity of arbitrary stream specifications.

\bibliography{ref}

\end{document}